\documentclass[onefignum,onetabnum]{siamart171218}
\newif\ifIEEE
\IEEEfalse

\usepackage{mathrsfs,amssymb,amsfonts,url,bbm,enumerate}
\usepackage{graphicx, subfigure, color}
\usepackage[draft]{fixme}
\usepackage{enumitem,dsfont}

\ifIEEE
\usepackage{amsmath,amsthm}

\theoremstyle{plain}
\newtheorem{theorem}{Theorem}
\newtheorem{lemma}{Lemma}

\newtheorem{corollary}{Corollary}

\newtheorem{claim}{Lemma}

\theoremstyle{definition}
\newtheorem{definition}{Definition}

\else 
\usepackage{lipsum}

\newsiamremark{remark}{Remark}
\newsiamremark{hypothesis}{Hypothesis}
\newsiamremark{example}{Example}
\newsiamthm{claim}{Claim}
\newsiamthm{conjecture}{Conjecture}

\author{Arsalan Sharifnassab\thanks{Department of Electrical Engineering, Sharif University of Technology, Tehran, Iran 
  (\email{a.sharifnassab@gmail.com}, \email{golestani@sharif.edu}).}
\and John N. Tsitsiklis\thanks{Laboratory for Information and Decision Systems, Electrical Engineering and Computer Science Department, MIT, Cambridge MA, 02139, USA 
  (\email{jnt@mit.edu}).}
\and S. Jamaloddin Golestani\footnotemark[2]}

\usepackage{amsopn}

\fi

\newcommand{\hide}[1]{}

\renewcommand{\tilde}[1]{\widetilde{#1}}

\newcommand{\grad}{\nabla}
\newcommand{\R}{\mathbb{R}}

\newcommand{\Ltwo}[1]{\big\lVert #1 \big\rVert}

\newcommand{\actionset}{\mathcal{S}}

\newcommand{\D}{D}

\newcommand{\conv}{\textrm{Conv}}

\newcommand{\region}{\mathcal{R}}
\newcommand{\degion}{\mathcal{D}}

\newcommand{\ball}{\mathcal{B}}

\newcommand{\Phit}{{\Phi}}

\newcommand{\ones}{\mathds{1}}
\newcommand{\I}{\mathcal{I}}
\newcommand{\pp}{\mathcal{P}}

\def\int{\rm{int}}

\long\def\red#1{{\color{red}#1}}
\usepackage[normalem]{ulem}
\long\def\del#1{{\color{green}\sout{[#1]}}}

\long\def\pur#1{{\color[rgb]{.8,0,.8}#1}}
\long\def\oli#1{{\color[rgb]{0,.8,.8}[#1]}}
\newcommand\redsout{\bgroup\markoverwith{\textcolor{red}{\rule[0.5ex]{2pt}{.5pt}}}\ULon}
\long\def\dela#1{{\color{blue}\redsout{[#1]}}}

\long\def\comm#1{{\color{blue} [#1]}}
\long\def\old#1{}

\title{Nonexpansive Piecewise Constant Hybrid Systems are Conservative
\ifIEEE
\else
\thanks{\today. This work was partially done while A. Sharifnassab was a visiting student at the Laboratory for Information and Decision Systems, MIT, Cambridge MA, 02139, USA.}
\fi}

\begin{document}

\maketitle

\begin{abstract}
Consider a partition of $\R^n$ into finitely many polyhedral regions $\degion_i$  and associated drift vectors $\mu_i\in\R^n$. We study ``hybrid'' dynamical systems whose trajectories have a constant drift, $\dot x=\mu_i$, whenever $x$ is in the interior of the $i$th region $\degion_i$, and behave consistently on the boundary between different regions. 
Our main result asserts that if such a system is \emph{nonexpansive} (i.e., if the Euclidean distance between any pair of  trajectories is a nonincreasing function of time), then the system must be 
\emph{conservative}, i.e., its trajectories are the same as the trajectories of the negative subgradient 
flow associated with a potential function. Furthermore, this potential function 
is necessarily convex, and is linear on each of the regions~$\degion_i$. 
We actually establish a more general version of this result, by making seemingly weaker assumptions on the dynamical system of interest.

\old{A continuous-time dynamical system is \emph{nonexpansive} if the Euclidean distance between any pair of its trajectories is a nonincreasing function of time; it is \emph{conservative} if its trajectories move along the gradient of a convex function; and it is \emph{finite-partition} if there is a partition of the domain into a finite number of regions, over each of which the drift is a constant vector.
While not every nonexpansive dynamical system is conservative, we show  that every nonexpansive finite-partition system is conservative.
More concretely, we show that under some minimal well-formedness assumptions, every nonexpansive finite-partition system is essentially 
the subgradient field of a \emph{convex} potential function, which is furthermore piecewise linear with finitely many pieces.}
\end{abstract}

\ifIEEE
\else
\fi

\section{\bf Introduction} \label{sec:introduction}

In this paper we study ``hybrid'' dynamical systems with the following special structure.  Consider a partition of $\R^n$ into finitely many polyhedral regions $\degion_i$, and associated drift vectors $\mu_i\in\R^n$. We focus on  systems whose trajectories obey $\dot x=\mu_i$ whenever $x$ is in the interior of the $i$th region $\degion_i$, and refer to them 
 as \emph{polyhedral hybrid systems}.

Polyhedral systems arise in many different contexts in which the dynamics are relatively simple, 
with
a finite set of possible control actions,and different actions   applied in different regions of the state space. 
Examples include communication networks \cite{TassE92,Neel10}, processing systems \cite{RossBM15}, manufacturing systems, and inventory management \cite{PerkS98,Meyn08}.
More concretely, this type of systems describes
the fluid model dynamics of several policies for real-time job scheduling \cite{Neel10} that choose at each time a service vector out of a finite set of possible such vectors, based on the current system state (i.e., the queue lengths), as in 
 the celebrated Max-Weight algorithm \cite{TassE92} and its generalizations \cite{GeorNT06}.
See 
Fig.~\ref{fig:application finite part is subdiff} for a simple example. 

\begin{figure*} 
\begin{center}
\includegraphics[width = .45\textwidth]{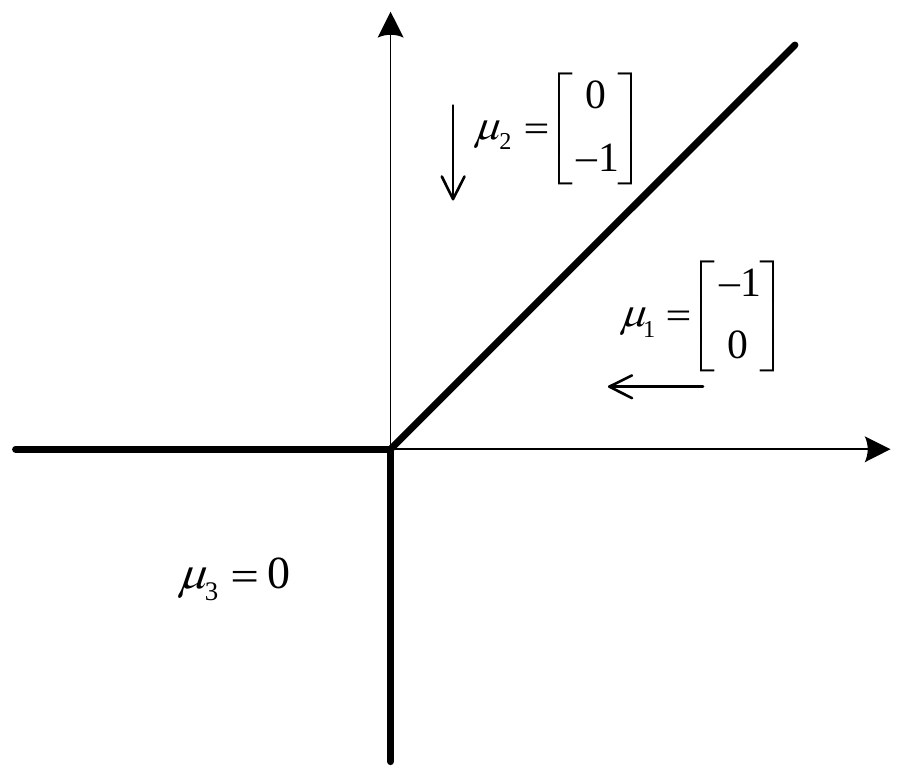}
\ifIEEE
\vspace{-1cm}
\fi
\end{center}
\caption{The dynamics of a network with two  queues and a server that serves a longest queue at each time, with arbitrary time-sharing in case of a tie. Moreover, if both  queues are negative, the server goes idle.
There are three polyhedral regions, and there is a constant drift vector, $\mu_i$, in the interior of each region. Furthermore at the boundary between two (or three) regions, the direction of motion is a convex combination of the drifts $\mu_i$ in the adjacent regions. Such a system is nonexpansive and can also be described as
the negative subgradient flow associated with the piecewise linear convex function $\Phi(x)=\max\big(-\mu_1^Tx,\, -\mu_2^Tx,\, 0\big)$.   
}
\label{fig:application finite part is subdiff}
\end{figure*}

Our main result asserts that if a polyhedral hybrid system  is \emph{nonexpansive} 
(i.e., if the Euclidean distance between any pair of trajectories is a nonincreasing function of time) then the system must be 
\emph{conservative}, i.e., its trajectories are the same as the trajectories of the negative subgradient 
flow associated with a potential function. Furthermore, this potential function can be chosen to be piecewise linear, with finitely many linear pieces, and convex; finally, the potential function is linear on each of the regions $\degion_i$.

Once we establish that the system is conservative, the particular properties of the potential function are not too surprising. For example, it is not hard to show that if a dynamical system is conservative and nonexpansive, then the underlying potential function must be convex.\footnote{We provide the proof for the special case of a conservative system $\dot x= -\nabla \Phi(x)$ associated with a \emph{differentiable} potential function $\Phi$.
Consider a pair $x$ and $y$ of points in $\R^n$, and  let  $x(\cdot)$ and $y(\cdot)$ be two trajectories initialized at $x$ and $y$, respectively. 
Then, the nonexpansive property implies that 
\begin{equation}
2\cdot (x -y )^T \big(\grad \Phi(x) -\grad \Phi(y)  \big)  \, = \, - \frac{d}{dt} \Ltwo{x(t)-y(t)}^2\, 
\geq\, 0.
\end{equation}
Thus, $\Phi$ has an increasing directional derivative over the line segment that starts at $y$ and ends at $x$.
Since $x$ and $y$ are arbitrary,  it follows that $\Phi$ is convex.}
The converse also turns out to be true: the (negative) subgradient flow associated with a  convex potential function is nonexpansive.
\footnote{For the special case where the convex  potential function $\Phi(\cdot)$ is differentiable, the argument is as follows. For any pair  $x(\cdot)$ and $y(\cdot)$ of trajectories and any time $t$, it follows from  the convexity of $\Phi(\cdot)$  that
\begin{equation}
 \frac{d}{dt} \Ltwo{x(t)-y(t)}^2 \,=\, -2\cdot \big(x(t) -y(t) \big)^T \Big(\grad \Phi\big(x(t)\big) -\grad \Phi\big(y(t)\big)  \Big) 
\le\, 0,
\end{equation}
and the system is nonexpansive.
}

The unexpected part of our result is that a nonexpansive polyhedral system is conservative in the first place.
This is quite surprising because a nonexpansive system is not necessarily conservative,\footnote{A simple example is a system in $\R^2$ whose trajectories are circles, traversed with a fixed angular velocity.}  and these two properties appear to be quite unrelated.
In general, a system being conservative is intimately tied with a curl-free property of the underlying field. On the other hand,  nonexpansive systems have no reason to be curl-free. It is thus remarkable that for the systems that we consider the combination of nonexpansiveness and the polyhedral structure  imposes a combinatorial structure around points where regions intersect, which then translates to a curl-free condition.

We note that we  actually establish our main result for a seemingly broader class of systems: well-formed nonexpansive finite-partition systems; cf.~Definitions \ref{def:well formed} and \ref{def:finite-partition},
without assuming, for example, that the different regions are polyhedral. It so happens that the nonexpansiveness assumption together with well-formedness forces the regions to be convex; and a partition of $\R^n$ into convex regions forces the regions to be polyhedral. 
Thus, this broader class of systems in fact reduces to the class of polyhedral hybrid systems introduced in the beginning of this section.


The rest of the paper is organized as follows.  
In Section \ref{sec:pwc}, we lay out the setting of interest and present the main result, 
whose proof  
is then given in Section \ref{sec:proof pwc}. 
In the course of the proof, we use several lemmas, whose proofs are relegated to the appendices for improved readability.
Finally, we discuss the implications of our result, along with some open problems, in Section~\ref{sec:discussion}.

\old{A much studied class of vector fields is the class of  \emph{conservative} vector fields, for which $F(\cdot)$ is the (negative) of the gradient field of some potential function.
If a conservative field is nonexpansive, then the potential function is necessarily convex.

On the other hand, a nonexpansive system is not necessarily conservative. 
For example, the linear dynamical system $\dot{x}=Ax$ with
\begin{equation}\label{eq:matrix A in the example}
A = \left[ \begin{array}{cc} -0.1 &-1 \\1 & -0.1 \end{array}\right],
\end{equation}
is nonexpansive. However, it cannot be the gradient field of any function, because $A$  is not symmetric.
See Fig. \ref{fig:rotary} for   two nonexpansive dynamical systems that are not conservative.

\begin{figure*} 
\begin{center}
\includegraphics[width = .45\textwidth]{fig_rotary.pdf}
\includegraphics[width = .45\textwidth]{fig_lin_spiral.pdf}
\ifIEEE
\vspace{-1cm}
\fi
\end{center}
\caption{Example trajectories of nonexpansive dynamical systems that are not conservative. a) Two-dimensional dynamical system $\left(\dot{x},\dot{y} \right) = (-y,x)$.
b) Two-dimensional linear system 
$\dot{x}= Ax$, with $A$ defined in \eqref{eq:matrix A in the example}.
}
\label{fig:rotary}
\end{figure*}

Nevertheless,
in this paper, we show that every nonexpansive \emph{finite-partition} system is guaranteed to be conservative an, in particular, the (negative) 
the subgradient field of some piecewise linear convex function with finitely many pieces. }

\hide{
\del{It is worth mentioning however that 
there exist counterexamples that show that 
such a strong sensitivity bound 
fails to hold for the larger class of all
nonexpansive conservative systems \cite{AlTG18p3}. In particular, the finiteness assumption on the number of regions is crucial for such bounds to hold.}

\comm{I wonder whether the preceding paragraph would be better placed in the discussion section, after the theorem has been stated. Also, a reader who reads this paper, without knowing about our larger agenda, may not understand why we care.}
\oli{As a motivation for significance of FPCS systems, maybe its good to keep some parts of this paragraph in the introduction. We can move the details (e.g., the last sentence) to the discussion section.}

\comm{Making a note to rethink the spin of this section in a later revision, after we post.}
}

\section{\bf Definitions and Main Result} \label{sec:pwc}
In this section, we define some terminology and three classes of systems. We then present our main result, which states that under a nonexpansiveness assumption, these three classes are equivalent.

\subsection{Definitions} \label{sec:defs}
We start with a definition of general dynamical systems, following 
\cite{Stew11}.

\begin{definition}[Dynamical Systems] \label{def:dynamical}
A \emph{dynamical system} is a set-valued function $F:\R^n\to 2^{\R^n}$.
A \emph{trajectory} of a system $F$ is a  function $x(\cdot)$ whose 
derivative, denoted by $\dot x$,  exists  and satisfies the differential inclusion
$\dot x(t) \in F\big(x(t)\big)$, except possibly for a measure-zero set of times.\footnote{The solution concept used here to described trajectories is consistent with the definition of (unperturbed) trajectories in \cite{AlTG19sensitivity}.
}
\end{definition}

For general systems, some assumptions are needed in order to guarantee existence  of trajectories, which leads us to the next definition. 

\begin{definition}[Well-Formed Dynamical System] \label{def:well formed}
A dynamical system $F$ is well-formed, if:
\begin{itemize}
\item[a)]  $F(\cdot)$ is upper semicontinuous.\footnote{A set-valued function $F$ is called upper semicontinuous if for  every $x\in \R^n$ and any open subset $\mathcal{V}$ of $\R^n$ such that $F(x)\subseteq \mathcal{V}$, there exists an open neighbourhood $\mathcal{U}$ of $x$ such that $F(y)\subseteq \mathcal{V}$, for all $y\in \mathcal{U}$.}


\item[b)] There exists a constant $\alpha$ such that $\Ltwo{F(x)}\le \alpha\big(1+\lVert x \rVert\big)$, for all $x\in \R^n$. 
\item[c)] For every $x\in\R^n$, $F(x)$ is a closed and convex set.
\end{itemize}
\end{definition}
The second condition above prevents solutions from blowing up in finite time.  
The other two conditions are technical, but are
often satisfied.

In this paper, we focus on nonexpansive systems, defined below.
\begin{definition}[Nonexpansive Systems] \label{def:nonexpansive}
A dynamical system $F$ is called \emph{nonexpansive} if for any pair  $x(\cdot)$ and $y(\cdot)$ of  trajectories 
and for any times $t_1$ and $t_2$, with  $t_2\ge t_1$,
\begin{equation}
\Ltwo{x(t_2)-y(t_2)}\le \Ltwo{x(t_1)-y(t_1)},
\end{equation} 
where $\|\cdot\|$ stands for the Euclidean norm.
\end{definition}

Note that the nonexpansiveness property automatically guarantees that if a solution exists, then it is unique. 
We next define the most general class of systems to be considered.

\begin{definition}[Finite-Partition Systems] \label{def:finite-partition}
A dynamical system $F$ is said to be a \emph{finite-partition} system 
if there  exists a finite collection of distinct vectors $\mu_1,\ldots,\mu_m$, such that for any $x\in\R^n$, 
there is some $i$ for which $\mu_i\in F(x)$.
\end{definition}

Loosely speaking, Definition \ref{def:finite-partition} requires that there be a finite set of ``special vectors'' $\mu$ so that at each point in the state space, at least one of these vectors is a possible drift, in the sense of $\dot x=\mu$. As we will see later, a finite-partition system induces a natural ``tiling'' of 
$\R^n$, a concept that we will define shortly. 
We will then proceed with a formal definition of the class of polyhedral hybrid systems, which was already discussed in Section~\ref{sec:introduction}. 
Throughout the paper, we use the term ``polyhedron'' to refer to a closed and convex  set that can be defined in terms of finitely many linear inequalities.


\begin{definition}[Polyhedral tiling]\label{eq:def:tiling}
A collection $\degion_i$, $i=1,\ldots,m$, of nonempty polyhedral subsets of $\R^n$ (``regions'')  is said to be a \emph{polyhedral tiling}
if it satisfies  the following:  
\begin{itemize}
\item [a)]
$\bigcup_i \degion_i=\R^n$; 
\item[b)]
each region $\degion_i$ has a nonempty interior;
\item[c)]
if $i\neq j$, then $\degion_i$ and $\degion_j$ have disjoint interiors.
\end{itemize}

\end{definition}

Given a polyhedral tiling, one can consider hybrid systems that have a constant drift in the interior of each polyhedron. In what follows, we use the notation ${\rm{Conv}(A)}$ to denote the convex hull of a set $A\subseteq\R^n$.

\begin{definition}[Polyhedral Hybrid Systems] \label{def:hybrid}\mbox{\hspace{-4pt}}
Consider a polyhedral tiling $\degion_1,$ $\ldots,\degion_m$, and associated \emph{distinct} vectors  $\mu_1,\ldots,\mu_m\in\R^n$. The corresponding \emph{polyhedral hybrid system} $F$ is defined by letting $F(x)={\rm Conv}\big(\{\mu_i \mid x\in\degion_i\}\big)$, for every $x\in\R^n$. In particular, if $x$ belongs to the interior of $\degion_i$, then $F(x)=\{\mu_i\}$.
\end{definition}

As will be discussed in Section~\ref{sec:discussion}, the requirement that the vectors $\mu_i$ are distinct can be made without loss of generality. 

The third class of systems that we will study is comprised of \emph{conservative} systems, that move along the negative subgradient of a potential field. 
When we further restrict the structure of the potential,  we obtain the class of systems that was studied in 
\cite{AlTG19sensitivity}.

\begin{definition}[FPCS Systems]\label{def:fpcs}
\begin{itemize}
\item[a)]
A function $\Phi:\R^n\to\R$ is called \emph{finitely piecewise linear and convex}  if it is of the form $\Phi(x)=\max_{i}\big(-\mu_i^Tx+b_i\big)$, 
for some finite set  of distinct pairs $(\mu_i,b_i)\in\R^n\times \R$, $i=1,\ldots,m$.
\item[b)]
We say that $F$ is a \emph{Finitely Piecewise Constant  Subgradient (FPCS)} system if there exists a finitely piecewise linear convex function $\Phi:\R^n\to\R$ such that for any $x\in\R^n$, $F(x)=-\partial\Phi(x)$,
where $\partial\Phi(x)$ denotes the subdifferential of $\Phi$ at $x$.
\end{itemize}
\end{definition}

FPCS systems are of particular interest, as they arise in various contexts. On the technical side they have some unusually strong sensitivity properties: as shown in 
\cite{AlTG19sensitivity}, the effect of an additive external perturbation on the state is bounded above by a constant times the magnitude of the \emph{integral} of the perturbation (as opposed to the integral of the magnitude, which is a much weaker upper bound). 

\subsection{Main Result} \label{sec:main}
We are interested in the relation between the following three classes of systems:
(i) well-formed finite-partition systems, (ii) polyhedral hybrid systems, and (iii) FPCS systems.  FPCS systems are automatically polyhedral hybrid systems, but the converse is not always true. Furthermore, finite-partition systems are much more general; in particular, they include systems where the regions of constant drift are not polyhedral. Nevertheless, our main result states that once we restrict our attention to nonexpansive systems, these three classes are essentially the same.

\hide{
\oli{Out of the following existence results, we only use Part (c) in the last paragraph of Subsection \ref{subsec:proof part d}. Can we consider removing the following 4 cases from here. Only briefly mentioning uniqueness (part a) at the end of the previous paragraph. And rendering Part (c) to the proof section (Subsection \ref{subsec:proof part d}) where it is used? }
Before continuing, let us dispense with the  question of existence and uniqueness, starting from any initial state. 
\begin{itemize}
\item[a)]
As long as we are dealing with nonexpansive systems, uniqueness of trajectories is automatic. 
\item[b)]
For the case of well-formed finite-partition systems, we have already noted that existence is guaranteed from Theorem~4.3. of \cite{Stew11}.
\item[c)] For the case of FPCS systems, existence follows from Corollary 4.6 of \cite{Stew11} and the fact (Lemma 2.30  of \cite{Stew11}) that any subgradient dynamical system  is a 
maximal monotone map\footnote{A set-valued function $F:\R^n\to 2^{\R^n}$ is a monotone map if for any $x_1,x_2\in\R^n$ and any $v_1\in F(x_1)$ and $v_2\in F(x_2)$, we have $\big(v_1-v_2\big)^T\big(x_1-x_2\big)\le 0$. It is called a maximal monotone map if it is monotone, and for any monotone map $\tilde{F}$, that satisfies $F(x)\subseteq \tilde{F}(x)$ for all $x$, we have $\tilde{F}=F$.}. \comm{Do we really need this argument? Can't we appeal to the result for finite-partition systems, as in the next item?} \oli{yes, immediate form Par (d) below and Part (c) of  the main theorem.}
\item[d)] It is not hard to check that a polyhedral hybrid systems is well-formed, and therefore  Theorem~4.3. of \cite{Stew11} again implies existence.  \comm{OK?} \oli{Right, but are you using Part (b) of  the main theorem below?}
\end{itemize}
}

\medskip


\begin{theorem}\label{th:pwc} Let $F$ be a nonexpansive dynamical system.
\begin{itemize}
\item[a)] $F$ is an FPCS system if and only if it is a polyhedral hybrid system.
\item[b)] If $F$ is a polyhedral hybrid system, 
or, equivalently, an FPCS system, 
then it is a well-formed finite-partition system.
\item[c)] If $F$ is a well-formed finite partition system, then there exists 
a polyhedral hybrid system (or, equivalently, an FPCS system) $F'$ such that for any initial state $x(0)$, the two systems follow the same trajectory. Furthermore, $F'(x)\subseteq F(x)$, for all $x\in\R^n$.
\end{itemize}
\end{theorem}

While some parts of the result are straightforward, Part (c) is
rather surprising. Nonexpansive systems are in general not conservative. Nevertheless,
nonexpansive polyhedral hybrid systems turn out to be be conservative 
because they are implicitly forced to have a special and 
subtle structure at the intersection of different regions. In the same spirit, 
nonexpansiveness, the finite-partition property, and well-formedness, taken together, force a polyhedral structure on the regions involved.

On the technical side, Theorem \ref{th:pwc} establishes equivalence of finite-partition and FPCS systems, in terms of the trajectories that they generate, even though the mappings $F$ may be a bit different. In order to appreciate what is involved here, consider a one-dimensional system of the following form: 
$$F(x)=\begin{cases}
1,& \mbox{if } x<0,\\
[-2,2],&  \mbox{if } x=0,\\
-1,&\mbox{if } x>0.
\end{cases}
$$
Trajectories of this system move at unit speed towards the origin; once at the origin, trajectories remain there. 
We observe that $F$ is upper semicontinuous and that
this is a well-formed nonexpansive finite-partition system, with two drift vectors, $\mu_1=-1$ and $\mu_2=1$. Yet, 
it is not a polyhedral hybrid system because $F(0)$ is strictly larger than the convex hull of the set $\{-1,1\}$ of drifts  in the  regions adjacent to 0. In the same spirit, the trajectories of $F$ are the negative subgradient flow for the convex potential function $\Phi(x)=|x|$, but, strictly speaking, $F$ is not a FPCS system, because
$F(0)=[-2,2] \neq [-1,1]=-\partial \Phi(0)$. On the other hand, the system $F$ is ``equivalent'' to the FPCS system associated with  $\Phi(x)=|x|$  (and also a polyhedral hybrid system), in the sense that they generate the same trajectories.

\hide{
It is easy to see that every subgradient system is well-formed and nonexpansive. 
The reverse however is not true. Fig. \ref{fig:rotary} illustrates a well-formed and nonexpansive system which is not a subgradient system. 

In this paper, we  show that if a well-formed and  nonexpansive dynamical system has essentially a finite set of drifts, then it will be a subgradient system.
More specifically:
}

\hide{
\section{\bf Approximating  Dynamical  Systems with  finite-partiotion Systems}
Here we show that every nonexpansive dynamical system over a compact domain can be approximated by a finite-partiotion system.
By an approximation, we mean the trajectories of the two systems stay arbitrarily close over a fixed time window, and while in the compact domain.
We will then show that, even when the initial system is nonexpansive, this approximating finite-partition system cannot be nonexpansive. 
In particular, for the dynamical system in Fig.~\ref{fig:rotary} (b), the exists no approximating nonexpansive finite-partition system in the above sense.
This is because the plane orthogonal to the gradient at some point will cut trough the trajectory emanating from that point. This cannot be the case for gradient fields where trajectories intersect with each level set only once.

\oli{Maybe it is better to remove this section, because it lessens the significance of nonexpansive finite-partition systems.}

\begin{corollary} [Approximation of Nonexpansive Systems by Subdifferential Systems]
Consider a well-formed and nonexpansive dynamical system $F$ on the unit ball of $\R^n$.
Then, for any $\epsilon, T>0$, there exists a subgradient system $-\partial\Phi$ such that 
for any pair of trajectories $x(\cdot)$ and $y(\cdot)$ of $F$ and $-\partial\Phi$, initiated at $x(0)=y(0)$, 
\begin{equation}
\Ltwo{x(t)-y(t)}   \le  \epsilon,
\end{equation}
for all $t\le T$.
\end{corollary}

- where else do the finite-partiotion systems arise? add references and applications
}



\section{\bf Proof of Theorem \ref{th:pwc}} \label{sec:proof pwc}

\hide{, organized in a sequence of three subsections. We first review the outline of the proof in Subsection~\ref{subsec:proof outline}  and 
 state a set of combinatorial and geometric lemmas in Subsection~\ref{subsec:proof pre}. 
We then present the proof of the theorem in Subsection \ref{subsec:proof th pwc}.}

In this section, we give the proof of Theorem~\ref{th:pwc}. 
Throughout, and for any statement that we make about various systems, we always assume that the systems are nonexpansive, even if this assumption is not explicitly stated.
\subsection{\bf Part (a): Polyhedral hybrid systems and FPCS systems are equivalent} 
We begin with the difficult direction of Part (a): we show that nonexpansive polyhedral hybrid dynamical systems are conservative.
The proof is given in a sequence of several lemmas, whose proofs are relegated to appendices. The general line of argument involves an explicit construction of a
convex potential function associated with a given polyhedral hybrid system. To develop our construction, we  first consider geometric polygonal paths that travel from one region to another. We endow these paths with certain weights and
show that cycles have zero weight. We then  leverage this conservation property to associate a weight with each region. We finally use these region weights to define
a convex potential function over $\R^n$, and
show that 
its subdifferential flow yields the same trajectories as the given polyhedral hybrid system.

\vspace{10pt}

We now start with the formal proof.
Consider a polyhedral tiling $\degion_1, \ldots, \degion_m$, 
distinct vectors $\mu_1,\ldots,\mu_m$, and the
 associated polyhedral hybrid system 
\begin{equation}\label{eq:Fx}
F(x)={\rm Conv}\big(\{\mu_i \mid x\in\degion_i\}\big).
\end{equation}
We say that two regions  $\degion_i$ and $\degion_j$,  are \emph{adjacent} if 
$i\neq j$ and
their intersection is nonempty. 
To any pair of adjacent regions $\degion_i$ and $\degion_j$ we associate a weight $b_{ij}$, by defining
\begin{equation}\label{eq:def of bij in proof of th pwc}
b_{ij}\triangleq \sup_{x\in \degion_i} \big(\mu_i-\mu_j\big)^T x.
\end{equation}
We also let $b_{ii} = 0$, for all $i$.

\hide{
\begin{definition}[Curl-Free Map-Graph] \label{def:curl free map g}
Consider a finite-segmentation $\degion=\left\{ \degion_i, \, i=1,\ldots,m  \right\}$ and its corresponding weighted map-graph $G$, endued with weights $b_{ij}$ for all links $ij$. Then, $\big(\degion, G \big)$ is said to be curl-free if  the following conditions hold:
\begin{itemize}
\item[1) ]
 $b_{ij}=-b_{ji}$, for all adjacent pair of regions $\degion_i$ and $\degion_j$,
\item[2) ]
$ b_{ij}+b_{jk}+b_{ki}=0$,  for all triples $i,j,k$ for which $\degion_i\bigcap\degion_j\bigcap\degion_k\ne\emptyset$.
\end{itemize}
\end{definition}
}


\begin{lemma}[Local Conservation of Weights] \label{lem:D is curl free}
The supremum in \eqref{eq:def of bij in proof of th pwc} is always attained, and every $b_{ij}$ is finite. 
Furthermore,
the weights $b_{ij}$ have the following properties:
\begin{itemize}
\item[a) ]
 $b_{ij}=-b_{ji}$, whenever $\degion_i$, $\degion_j$ are adjacent;
\item[b) ]
$ b_{ij}+b_{jk}+b_{ki}=0$,  for all triples $i,j,k$ for which $\degion_i\bigcap\degion_j\bigcap\degion_k\ne\emptyset$.
\end{itemize}
\end{lemma}
The proof is given in Appendix~\ref{subsec:G is curl free} and relies on the  nonexpansive property of the dynamics. 

In our next step, we associate to each region $\degion_i$ a weight $b_i$, such that $ b_i-b_j = b_{ij}$, for all pairs of adjacent regions $\degion_i$ and $\degion_j$.
This is made possible because of a global conservation property of weights over geometric paths on the tiling.
Before going through this global conservation property, we need an auxiliary lemma and some definitions. Lemma \ref{lem:prop of tiling of Rn} below essentially asserts that if we can find points that come sufficiently close, simultaneously, to each one of three polyhedra, then these polyhedra must have a common element.
Note that such a property is special to polyhedra, and does not hold more generally: it is not hard to find examples of disjoint closed and convex sets whose distance is zero.

\begin{lemma}[A Geometric Property of Polyhedral Tilings]\label{lem:prop of tiling of Rn}
Consider the regions $\degion_i$ in a polyhedral tiling.
There exists a constant $\gamma>0$ such that if a closed Euclidean ball of radius $\gamma$ intersects with any (not necessarily distinct) three regions $\degion_i$, $\degion_j$, and $\degion_k$, then $\degion_i \bigcap \degion_j \bigcap \degion_k$ is nonempty.
\end{lemma}
The proof is given in Appendix~\ref{app:proof 3ball}, and relies on a bound derived through an auxiliary linear program.  
We now fix the constant $\gamma$ of Lemma~\ref{lem:prop of tiling of Rn}, and let
\begin{equation} \label{def:delta}
\delta = \gamma/{3}.
\end{equation}

\begin{definition}[Paths and Cycles] \label{def:path}
A sequence $x_1,\ldots,x_t$ of points of $\R^n$ is a {\bf path} if each $x_i$ is in the interior of some region $\degion_{k_i}$. 
A path is {\bf jump-free} if for $i=1,\ldots,t-1$, the regions $\degion_{k_i}$ and $\degion_{k_{i+1}}$ are either the same or adjacent.
A path is {\bf fine} if for $i=1,\ldots,t-1$, and with $\delta$ as defined in \eqref{def:delta}, we have $\Ltwo{x_i-x_{i+1}}\le \delta$.
To every jump-free path, we associate a weight: 
\begin{equation}
W\big(x_1,\ldots,x_t\big) \,=\, \sum_{i=1}^{t-1}  b_{k_i k_{i+1}} .
\end{equation}
A {\bf cycle} is a path that ``ends where it started'', i.e.,
 $x_1=x_t$.
\end{definition}

Note that, as a consequence of Lemma~\ref{lem:prop of tiling of Rn}  and our choice of $\delta$, every fine path is jump-free.\footnote{To see this, consider $x_i$ and $x_{i+1}$ on a fine path. The segment that joins $x_i$ to $x_{i+1}$ is contained in a ball of radius at most $\delta$, and which intersects both  $\degion_{k_i}$ and $\degion_{k_{i+1}}$. Then, 
Lemma~\ref{lem:prop of tiling of Rn} applies and shows that  $\degion_{k_i}$ and $\degion_{k_{i+1}}$ must intersect, and are therefore the same or adjacent.}
The next lemma establishes a global conservation property for path weights.
\begin{lemma}[Global Conservation of Weights] \label{lem:map graph}
Every fine cycle has zero weight.
\end{lemma}
The proof is given in Appendix~\ref{app:proof stokes}. In the proof, we first leverage Lemmas~\ref{lem:D is curl free} and~\ref{lem:prop of tiling of Rn} to show that every   fine cycle of ``small size'' (i.e., contained inside a small ball) has zero weight. We then
represent  a given fine cycle  as  a superposition of multiple fine cycles of smaller  sizes, and use  induction.


Using Lemma~\ref{lem:map graph}, we can now associate a weight $b_i$ to each  region $\degion_i$, as follows. 
Let $b_1=0$. For any other  region $\degion_i$, let $b_i$ be the weight of an arbitrary  fine path 
$x_1,\ldots,x_t$, with $x_1\in \degion_i$ and $x_t\in \degion_1$.
Then, using Lemma~\ref{lem:map graph}, it is not hard to see that
$b_i$ is independent of the choice of the fine path, and is therefore well-defined.
Moreover, for any pair of adjacent regions $\degion_i$ and $\degion_j$, we have 
\begin{equation} \label{eq:bij is bi-bj}
b_{ij} = b_i-b_j.
\end{equation}

\hide{
\comm{I have not yet looked at the proof of the next lemma, but the story looks convoluted. I suspect that similar arguments are repeated in proving parts (a) and (b).
I think that a natural way to go is to first define
$$\Phit(x)=\max_j\, (- \mu_j^Tx+b_j )$$ and then state the lemma as
$$\max_j\, (- \mu_j^Tx+b_j ) = - \mu_i^Tx+b_i$$
if and only if $x\in\degion_i$. (In words, $\degion_i$ is precisely the set of $x$ on which $\mu_i$ attains the maximum. If you agree, can you draft a proof in the most economical way? All of the proof is to be in the appendix.} \oli{OK. Please see the new lemma and the new shorter proof below.}

\comm{NOTE: It is not clear whether the ``max'' is/should be taken over all $i$ or just over those $i$ for which $\degion_i$ is nonempty. Of course, the two are equivalent, but this needs a short argument, and can affect the flow of the proof of the lemma. Some statements, e.g., part (b) of the lemma are true only if it is a restricted maximum.} \oli{Right. But we are now in  part (a) of the theorem. Note that $\I$ belongs to Part (d).}
}
We now define a potential function $\Phi:\R^n\to\R$,  by letting
\begin{equation} \label{eq:def phi based on -mu , b}  
\Phit(x)=\max_{i} \, (- \mu_i^Tx+b_i ).
\end{equation} 

\begin{lemma}[Properties of $\Phi$] \label{lem:Phi is well defined and convex}
For any $x\in\R^n$, we have $\Phit(x) =  -\mu_i^Tx+b_i$ if and only if $x\in\degion_i$.  
\end{lemma}

The proof is given in Appendix~\ref{subsec:phi is convex} 
and relies on a delicate interplay between the definition of weights $b_i$ and $b_{ij}$ and the properties of polyhedral tilings. 
Furthermore, it also relies on the  requirement that the vectors $\mu_i$ be distinct; cf.~Definition~\ref{def:hybrid}; the lemma fails to hold otherwise.

Let us fix some $x\in\R^n$. 
Since $\Phit$ is convex, it has a subdifferential at $x$, and 
\begin{equation} \label{eq:subdiff phit in F}
-\partial \Phit(x) \,=\,   \conv\left\{\mu_i\,\big|\,  i\in \operatorname{argmax}_j \, (- \mu_j^Tx+b_j )\right\}    \,=\,  \conv\left\{\mu_i \mid x\in \degion_i\right\} \,=\, F(x).
\end{equation}
where the first equality is from the subdifferential formula for the pointwise maximum functions \cite[Section 3.1.1]{Bert15}, the second equality follows from Lemma~\ref{lem:Phi is well defined and convex}, and the last equality is  the definition of $F$. 
Therefore, the polyhedral hybrid system $F$ coincides with the FPCS system associated with the particular potential function $\Phit$ that we introduced.

This completes the proof of one direction in Part (a) of the Theorem.
The proof of the converse statement, that FPCS systems are polyhedral hybrid systems, is fairly straightforward. It only requires us to check a few details, which is done  
 in Appendix~\ref{sec:easy-proof}.

\hide{
\comm{Somehow prove that $\degion_i$, $i\in I$ covers all of $\R^n$.
\\
Then prove that if $i\neq j$ and $i,j\in I$, then $\degion_i$, $\degion_j$ have disjoint interiors.
\\
This proves that $\degion_i$, $i\in I$ is a polyhedral tiling.}

\comm{Argue that 
$$\Phit'(x)=\max_{i\in I} \, (- \mu_i^Tx+b_i )$$
is the same as $\Phit(x)$. This can be done by showing that in the interiors of the regions $\degion_i$, the two functions coincide. Using continuity, the two functions must coincide everywhere.}

\comm{Argue that tiling $\degion_i$, $i\in I$, together with the vectors $\mu_i$, give a polyhedral hybrid system that coincides with the subgradient of $\Phi'_t$}

\dela{Some of the statements in  Theorem  \ref{th:pwc} are straightforward. We dispense with them first.}
}
\medskip	
\subsection{\bf Part (b): Polyhedral hybrid (or FPCS) systems are well-formed finite-partition systems}
This result is just a simple observation. The details are as follows. 
Consider a polyhedral hybrid system $F$,  with $F(x) = \conv\big(\left\{\mu_i \mid x\in \degion_i\right\}\big)$,
and fix some particular $x$. 
Recall that every polyhedral region $\degion_j$ is closed and that $\bigcup_j \degion_j=\R^n$.
Therefore, 
$\mathcal{U} \triangleq \R ^n \backslash \bigcup_{j \,:\, x\not\in\degion_j}\degion_j$ is an open neighbourhood of $x$ contained in $\bigcup_{i \,:\, x\in\degion_i}\degion_i$. 
It follows that for every  $y\in \mathcal{U}$, we have  $\{\mu_i \mid y\in\degion_i\} \subseteq \{\mu_i \mid x\in\degion_i\}$, and as a result,
$F(y) = {\rm Conv} \big( \{\mu_i \mid y\in\degion_i\} \big)
\subseteq {\rm Conv}\big(\{\mu_i \mid x\in\degion_i\}\big) = F(x)$. 
Therefore,
 $F$ is upper semicontinuous. 
The bound on $\|F(x)\|$ is immediate because $F(x)$ is a subset of the convex hull of $\mu_1,\ldots,\mu_m$.
Finally, each $F(x)$ is the convex hull of finitely many vectors, and is therefore closed and convex. It follows that the system is well-formed. Furthermore, each $F(x)$ contains at least one of the vectors $\mu_i$, which implies that we have  a well-formed finite-partition system. 

Given that FPCS systems are equivalent to hybrid systems, FPCS systems are also well-formed finite-partition systems, and this concludes the proof of Part (b).


\subsection{Part (c): A finite-partition system has the same trajectories as a polyhedral hybrid system} \label{subsec:proof part d}

In this subsection, we  show that
any well-formed nonexpansive finite-partition system $F$ is associated with  a polyhedral hybrid system $F'$ that has the  same trajectories. The key step involves showing that under the nonexpansiveness assumption, the regions associated with a finite-partition system are essentially polyhedral.

We fix a well-formed and nonexpansive
 finite-partition dynamical system $F$ and an associated finite collection of  distinct vectors $ \mu_1,\ldots, \mu_m $  such that for any $x\in\R^n$, there exists some $\mu_i$ for which $\mu_i\in F(x)$.
 We define
a collection of regions,
\begin{equation} \label{eq:def region i}
\region_i \triangleq \left\{ x\in\R^n\,\big|\,  \mu_i\in F(x) \right\},\qquad i=1,\ldots, m.
\end{equation}
Our first lemma asserts that these regions are closed.

\begin{lemma} \label{lem:Ri closed}
The set $\region_i$ is closed, for every $i$.
\end{lemma}
The proof is based on the well-formedness of $F$, and is given in Appendix~\ref{sec:proof lem Ri closed}. 
We now consider another collection $\degion_1\,\ldots,\degion_m$ of regions, defined by
\begin{equation}\label{eq:def degion i in the proof of th pwc}
\degion_i\triangleq\textrm{closure}\big(\region_i^\circ\big), \quad i=1,\ldots,m,
\end{equation}
where $\region_i^\circ$ is the interior of $\region_i$. Then, Lemma~\ref{lem:Ri closed} implies that 
\begin{equation}\label{eq:degin subset region}
\degion_i\subseteq \region_i, \quad i=1,\ldots,m.
\end{equation}
The regions $\region_i$ can be fairly unstructured, For example, they may have multiple isolated points. On the other hand, the regions $\degion_i$ are much better behaved, as stated in the next lemma. We let $\mathcal{I}$ be the set of indices for which $\degion_i$ is nonempty (equivalently, $\region_i$ has nonempty interior).

\begin{lemma}[Regions $\degion_i$ Form a Polyhedral Tiling] \label{lem:Di is finite seg}
The regions  $\degion_i$, $i\in \mathcal{I}$, defined in \eqref{eq:def degion i in the proof of th pwc}, form a polyhedral tiling. 
\end{lemma}
The proof is given in Appendix~\ref{subsec:proof Di is finite seg}. 

\hide{
To any tiling we associate a \emph{map-graph}:
	
\begin{definition}[Map-Graphs] \label{def:map graph}
Given a tiling $\degion_i$, $i=1,\ldots,m$, of $\R^n$, 
the associated map-graph is  a directed graph of $m$ nodes, in which a pair $(i,j)$ of nodes have a directed link $ij$ in between if $\degion_i\bigcap\degion_j \ne \emptyset$.
\end{definition}
Let $G$ be the map-graph of the tiling $\degion$  in \eqref{eq:def fn D}.
 To each directed  link $ij$ of $G$ we associate a weight $b_{ij}$, 
}

We now define a system $F'$, by letting,  $F'(x) = \conv\left\{\mu_i\,\big|\, x\in \degion_i, \,i\in\I\right\}$, for every $x$.   From Lemma~\ref{lem:Di is finite seg},  the regions $\degion_i$ form a polyhedral tiling. As a result, $F'$ is a polyhedral hybrid system.
On the other hand, \eqref{eq:degin subset region} implies that  if $x\in \degion_i$, then
$x\in\region_i$, i.e., 
 $\mu_i\in F(x)$. 
Since $F$ is well-formed, $F(x)$ is convex. 
Therefore, $F(x)$ contains the convex hull $\conv\left\{\mu_i\,\big|\, x\in \degion_i\right\}$,
and we have established that $F'(x)\subseteq F(x)$, for all $x\in\R^n$, as desired.

To complete the proof of Part (d), it remains to show that the two systems, $F$ and $F'$, have the same sets of trajectories.
Since $F'(x)\subseteq F(x)$, for all $x$, it follows that any trajectory of $F'$ is also a  trajectory of $F$. 
For the reverse direction, let $x(\cdot)$ be a trajectory of $F$, initialized at some $x(0)$. 
Let   $y(\cdot)$ be a trajectory of $F'$ with the same initial condition, $y(0)=x(0)$. 
(General existence results for well-formed systems guarantee that such a trajectory $y(\cdot)$ exists; cf.~ Theorem~4.3 of \cite{Stew11}\pur{)}.
The inclusion $F'\subseteq F$ implies that $y(\cdot)$ is also a trajectory of $F$.  
However, since $F$ is nonexpansive,  it has a unique trajectory with initial condition $x(0)$.  
Therefore, $y(t)=x(t)$, for all $t\ge 0$, and $x(\cdot)$ is also a trajectory of $F'$.
This completes the proof of Part (c) of the theorem.

\hide{
The high level idea is to use the well-formedness and nonexpansive properties  to show that the closure of the interior of each region in the partition is a convex set.
Each region  is associated with a vector $\mu_i$ in $\actionset$ introduced in Definition \ref{def:finite-partition}. 
We then show that there exist constants $\{b_i\}$, such that $\Phi(x)=\max\big(-mu_i^T x + b_i \big)$ satisfies (\ref{eq:finite-partiotion has subdiff inside}). To do this, we consider a weighted \emph{map graph} of the regions in the partition, 
 and prove a combinatorial analog of Green's theorem from calculus, for the case of weighted map-graphs. We then use this result to find a potential function over the map graph, which is finally used to define coefficients $\{ b_i\}$ and construct the convex potential function $\Phi$.
 }
\medskip


\section{\bf Discussion}   \label{sec:discussion}

In this section, we discuss our result and its implications.
We also discuss some open problems and directions for future research.

\subsection{Our Result} \label{s:result}
We have established that nonexpansive polyhedral hybrid systems are conservative, and in particular they follow the (negative)
subgradient flow of a piecewise linear convex potential function, with finitely many pieces. We also showed that the same is true for a seemingly more general class of finite-partition systems.
One consequence of our results is that previously established upper bounds on the sensitivity of FPCS systems to additive external perturbations \cite{AlTG19sensitivity} now carry over to nonexpansive polyhedral hybrid systems. 

The finiteness of the number of constant-drift regions is central to this result, because 
in its absence, a well-formed nonexpansive dynamical system need not be conservative (cf.~Footnote~3).
In this respect, our result  is quite counter-intuitive: every smooth dynamical system over a compact domain can be approximated by 
a polyhedral hybrid system, 
with arbitrarily high accuracy over a bounded time interval. Hence, our result might suggest that every nonexpansive smooth system must be conservative, but this is not the case. 
This apparent contradiction is resolved by observing that such a finite approximation will not in general  preserve the nonexpansiveness property.

Our result  is analogous to a corollary of Stokes' theorem  \cite{Stew07}, that any curl-free vector field is conservative.
In our context, the curl-free property is replaced by a local conservation property  
of nonexpansive polyhedral hybrid systems; cf.~Lemma~\ref{lem:D is curl free}(b). 
 What is somewhat surprising however is that there is nothing in our assumptions that suggests  such a curl-free property. Instead, this property emerges through the delicate interplay  between the finiteness and the nonexpansiveness assumptions. 
In the same spirit, the polyhedral shape of the different regions in a finite-partition system is not assumed but emerges in an unexpected way from seemingly unrelated assumptions.

\hide{
\subsection{Implications}

Other than deepening our understanding of finite-partition dynamical systems, the proposed  reduction to the simpler and well-studied subgradient systems facilitates the analysis of nonexpansive finite-partition systems.
In particular, previously established upper bounds on the sensitivity of FPCS systems to additive external \del{disturbances} \red{perturbations} \cite{AlTG19sensitivity} now carry over to finite-partition systems, under some natural assumptions such as uniqueness of \del{disturbed} \red{perturbed} solutions.

Another implication of Theorem~\ref{th:pwc} concerns the approximability of nonexpansive dynamical systems by nonexpansive finite-partition systems.
Given a dynamical system $F$, we say that a  system $G$ is an $(\epsilon,T)$-approximation of $F$ if for any trajectory $x(\cdot)$ of $F$, there exists a trajectory $y(\cdot)$ of $G$ such that 
\begin{equation}
\Ltwo{x(t)-y(t)}\le \epsilon,  \qquad \forall t\le T.
\end{equation}
It is not difficult to see that for any nonexpansive dynamical system $F$ over a compact domain, and any $\epsilon,T > 0$, there exists a finite-partition system $G$ that provides  an $(\epsilon,T)$-approximation of $F$.
However, as already hinted in Section \ref{s:result}, 
 such a system $G$  cannot  be guaranteed to be nonexpansive, as this would lead to a contradiction.
\hide{Fig. \ref{fig:approximation} shows a nonexpansive dynamical system $F$ over the unit Euclidean ball. For this system,   and for sufficiently small values of $\epsilon$ and sufficiently large values of $T$,  no $(\epsilon,T)$-approximation of $F$ can be both finite-partition and nonexpansive.}
In summary, every nonexpansive dynamical system over a compact domain admits approximating finite-partition systems of arbitrary accuracy,
while there are examples of nonexpansive dynamical systems that do not admit arbitrarily close approximations by \emph{nonexpansive} finite-partition systems.

\hide{
\begin{figure*} 
\begin{center}
\includegraphics[width = .45\textwidth]{fig_approx.pdf}
\ifIEEE
\vspace{-1cm}
\fi
\end{center}
\caption{A trajectory of the nonexpansive dynamical system $\dot{x}=Ax$ with $A$ defined in \eqref{eq:matrix A in the example}.
The figure shows a pair $x(0)$ and $x(T)$ of points on the trajectory, a line $l$ orthogonal to the trajectory at  $x(0)$, and its corresponding half-plane $W$ (the highlighted region).
If $F$ was a subgradient field, then $W$ would have contained the sub-level-set of the potential function that touched $x(0)$. 
Then, $x(t)$ would have been contained in $W$, for all $t\ge 0$.
However, $x(T)$ lies outside of $W$.
It follows that  $F$ cannot be a subgradient field.
Moreover, for sufficiently small $\epsilon$, $F$ does not have any  $(\epsilon , T)$-approximating subgradient system.
Then, Theorem~\ref{th:pwc} implies that there exists no $(\epsilon , T)$-approximating nonexpansive finite-partition system.
}
\label{fig:approximation} 
\end{figure*}
}
}


\subsection{Open Problems}

Our result  suggests several possible generalizations and open problems for future research,
which we list below.
\begin{itemize}
\item[a)] Does the result generalize to the case where the set $\actionset$ in the definition of finite-partition systems (Definition \ref{def:finite-partition}) is allowed to be countably infinite?
Or if we assume that we have a polyhedral tiling with a countably infinite number of regions $\degion_i$, would some additional assumptions be needed?

\item[b)] {What if we relax the constant drift property in the interior of each region?}
More concretely, consider a nonexpansive dynamical system $F$, and a polyhedral tiling of  $\R^n$  and assume that the restriction of $F$ to each  region is a (negative) subgradient field. Does this imply that $F$ is a (negative) subgradient field?
\item[c)] What if we consider a finite-partition system defined only on a convex (but not necessarily polyhedral) subset of $\R^n$? We conjecture that the result remains valid, but a different proof seems to be needed. 
\end{itemize}
\hide{
\oli{Let me think more about these problems for the next revision to make sure that they  don't have simple solutions.}

\oli{One promising condition is: ``every compact set intersects with a finite number of regions''. I believe that our theorem is good under this condition.}

\oli{I also need to review the relevant literature, if any...}

\comm{Sentence removed from the intro, that could be inserted later in the discussion:
It is worth mentioning however that 
there exist counterexamples that show that 
such a strong sensitivity bound 
fails to hold for the larger class of all
nonexpansive conservative systems \cite{AlTG19mis}. In particular, the finiteness assumption on the number of regions is crucial for such bounds to hold.}
}

\old{
\comm{For our paper [8], let us give pointer to Arxiv instead. And add ``submitted''} \oli{Done! Also added ``TAC''. OK?}
\comm{I do not see "TAC",  I am probably working with an outdated .bib file.\\
Some stylistic comments on references. \\
Journal titles should be roman. \oli{I have no control over that.} Book titles and journal titles should be italic and capitalized.\\ 
In [1], remove "Belmont\\
In [2], the title of the paper should be lower case. \oli{But [2] is a book...}\\
"jit" should be "JIT"} \oli{Updates the bib file to comply with the above standards.}}



\bibliographystyle{siamplain}
\bibliography{schbib}


\ifIEEE
\newpage
\appendices
\else
\appendix
\medskip
\medskip
\medskip
\fi

\old{
\section{\red{The $\mu_i$ can be assumed distinct}}\label{lem:distinct}
\red{Consider a polyhedral hybrid system $F$, specified by regions $\degion_i$ and corresponding drifts $\mu_i$, $i=1,\ldots,m$. According to Definition~\ref{def:hybrid}, some of the drifts $\mu_i$ could be the same. In this appendix, we establish that this assumption can be removed without loss of generality. In particular, we will show that there exists a polyhedral hybrid system $F'$, specified by a polyehdral tiling $\degion'_1,\ldots,\degion_m'$ and corresponding drifts $\mu'_1,\ldots,\mu'_{m'}$, such that: (i) $F(x)=F'(x)$, for all $x$, (ii) The drifts $\mu'_i$ are distinct.

The argument is as follows. Let $\mu'_1,\ldots,\mu_{m'}$ be the distinct values of the drifts $\mu_i$ in the original system $f$. For $j=1,\ldots,m'$, let 
$$\degion'_j=\bigcup_{i: \mu_i=\mu_j'}\degion_i.$$
We first establish that the new regions form a polyehdral tiling. The union of the regions $\degion'_j$ is the same as the union of the regions $\degion_i$, and is therefore all of $\R^n$. Each region $\degion'_j$ contains at least one region $\degion_i$; since the regions $\degion_i$ have nonempty interior, the same is true for the regions $\degion_j'$. Finally, 
suppose that the interiors of two regions $\degion'_j$, $\degion'_l$, with $l\neq j$, have a common element $x$. By perturbing $x$, we can obtain a new element $x'$ which is in the interior of both $\degion'_j$ and $\degion'_l$, and which does not belong to the (lower-dimensional) boundary of any of original polyhedra $\degion_i$. But this can happen only if $x'$ belongs to the interior of two different original regions $\degion_i$, which is impossible because of condition (c) in  Definition~\ref{eq:def:tiling}.

We have so far verified conditions (a)-(c) in Definition~\ref{eq:def:tiling}. It remains to show that the regions $\degion'_j$ are polyhedral. Since these regions are unions of polyhedra $\degion_i$, we only need to show that each $\degion'_j$ is convex. Suppose not. Then, there exists some $j$ and some distinct $x,y$ in the interior of $\degion'_j$ such that $(x+y)/2$ is not in $\degion_j'$. By slightly perturbing $x$ and $y$, we can and will assume that $x$, $y$, and $(x+y)/2$ are interior elements of the original regions $\degion_i$. Assume, without loss of generality that $x\in\degion_1$, $y\in\degion_2$, and $(x+y)/2\in\degion_3$. Since $x$ and $y$ belong to the same region $\degion'_j$, it follows that there exists some $\mu$ such that $\mu\in F(z)$ for all $z$ close enough to either $x$ or $y$. 
\comm{Please complete the proof - unless you have a shorter way of writing it up.}	
\comm{Perhaps this can go a bit later, as long as it before the appendix where it is used. This way, could maybe use some inequalities that are proved in between - 
}}
}

\section{\bf Proof of Lemma \ref{lem:D is curl free}} \label{subsec:G is curl free}
We consider a polyhedral hybrid system $F$, as in \eqref{eq:Fx}, and the weights defined in \eqref{eq:def of bij in proof of th pwc}. We start with the proof of Part (a) of the lemma.

Let $\degion_i^\circ$ stand for the interior of $\degion_i$, which is assumed nonempty for all $i$; cf.~Definition~\ref{eq:def:tiling}(b). 
Let $x_i\in\degion_i^\circ$ and $x_j\in\degion_j^\circ$,
 and let $x_i(\cdot)$  and $x_j(\cdot)$ be two trajectories with initial conditions $x_i(0)=x_i$ and $x_j(0)=x_j$.  
Then, at time $t=0$, we have $\dot{x}_i(0) = \mu_i$ and $\dot{x}_j({0}) = \mu_j$. Using the nonexpansive property of the dynamics, we obtain
\begin{equation} \label{eq:moved}
\begin{split}
\big(\mu_i-\mu_j\big)^T\big(x_i-x_j\big) \, &=\, \big(\dot{x}_i(0) -\dot{x}_j(0) \big)^T\big(x_i(0)-x_j(0)\big)\\
&  =\, \frac{1}{2}\cdot\frac{d^+}{dt} \Ltwo{x_i(t)-x_j(t)}^2 \Big|_{t=0}\\
& \le\, 0.
\end{split}
\end{equation}
By taking the limit in \eqref{eq:moved} along sequences in  $\degion_i^\circ$ and $\degion_j^\circ$ that converge to
$x_i$ and $x_j$, respectively, we conclude that~\eqref{eq:moved} holds for every $x_i\in\degion_i$ and $x_j\in\degion_j$. 

Consider two adjacent regions $\degion_i$ and $\degion_j$, and a point $x$ in their intersection.  
Since $x\in\degion_j$, it follows from \eqref{eq:moved} that if $x_i\in\degion_i$, 
then $\big(\mu_i-\mu_j\big)^Tx_i\le  \big(\mu_i-\mu_j\big)^Tx$.
Therefore, 
\begin{equation}
b_{ij} \,=\, \sup_{x_i\in\degion_i} \big(\mu_i-\mu_j\big)^Tx_i  \,\le\,\big(\mu_i-\mu_j\big)^Tx
<\infty.
\end{equation}
On the other hand, since $x\in\degion_i$, 
\begin{equation}
b_{ij} \,=\, \sup_{x_i\in\degion_i} \big(\mu_i-\mu_j\big)^Tx_i  \,\ge\,\big(\mu_i-\mu_j\big)^Tx.
\end{equation}
Therefore, for any $x\in\degion_i\bigcap\degion_j$,
\begin{equation} \label{eq:mu dif x at boundary equals bij}
b_{ij} \,=\, \big( \mu_i -\mu_j \big) ^T x.
\end{equation}
Switching the roles of $i$ and $j$, we obtain
\begin{equation} \label{eq: bij and bji ineq}
b_{ji} \,=\, \big( \mu_j -\mu_i \big) ^T x\,  =\, -\big( \mu_i -\mu_j \big) ^Tx \, = \, -b_{ij}.
\end{equation} 
This establishes the first part of the lemma.

For the second part,
consider a triple $i,j,k$, such that $\degion_i\bigcap \degion_j\bigcap \degion_k\ne\emptyset$. 
If two of these indices are equal, e.g., if $i=j\neq k$, then the first part of the lemma implies that
$$b_{ij}+b_{jk}+b_{ki}=b_{ii}+b_{ik}+b_{ki}=b_{ik}+b_{ki}=0.$$
If the three indices are distinct, then fix some
$x\in \degion_i\bigcap \degion_j\bigcap \degion_k$. It follows from (\ref{eq:mu dif x at boundary equals bij}) that
\begin{equation}
b_{ij}+b_{jk}+b_{ki}\, =\, \big(\mu_i-\mu_j\big)^Tx + \big(\mu_j-\mu_k\big)^Tx + \big(\mu_k-\mu_i\big)^Tx \,=\, 0.
\end{equation}
This completes the proof of the lemma.
\hfill \qed

\medskip


\section{\bf Proof of Lemma \ref{lem:prop of tiling of Rn}} \label{app:proof 3ball}
Consider three regions $\degion^{1}$, $\degion^{2}$, and $\degion^{3}$, with empty intersection, i.e., $\degion^{1}\bigcap \degion^{2}\bigcap \degion^{3}=\emptyset$.
Since the regions $\degion_i$, $i=1,2,3$, are polyhedra, there exist matrices
$A^i$, all rows of which have unit norm, and vectors $b^i$ such that 
\begin{equation}
\degion^i = \big\{x\in \R^n\,\big|\, A^i x-b^i \preceq 0\big\},\qquad i=1,2,3,
\end{equation}
where $\preceq$ stands for componentwise inequality.

Consider the linear programming problem with variables $\epsilon$ and $x$, of minimizing $\epsilon$ subject to
\begin{equation} \label{eq:lp}
A^ix-b^i\preceq \epsilon\ones^i, \quad i=1,2,3,
\end{equation}
where $\ones^i$ is a vector with the same dimension as $b^i$ and with all entries equal to one.
For any $x\in\R^n$, there is a large enough $\epsilon$ such that $(\epsilon,x)$ is a feasible point of \eqref{eq:lp}. 
Therefore,  \eqref{eq:lp} has a nonempty set of feasible points.
Moreover, since $\degion^{1}\bigcap \degion^{2}\bigcap \degion^{3}$ is empty, for any feasible point $(\epsilon,x)$ of \eqref{eq:lp}, we have $\epsilon>0$.
Given that we are dealing with a linear programming problem with a nonempty feasible region and 
with 
finite optimal cost, the optimal value of the objective is attained, and must be a positive number, which we denote by $\epsilon^*$.
Hence,
 for any $x\in\R^n$, the constraint $A^ix-b^i\preceq (\epsilon^*/2) \ones^i$ is violated, for at least one $i\in\big\{1,2,3\big\}$.
Since the rows of $A$ have unit norm, it follows that the closed $(\epsilon^*/2)$-neighbourhood of $x$ does not intersect $\degion^i$.
The desired result  follows by letting $\gamma$ be equal to the minimum of the constants $\epsilon^*/2$ over all triples $\degion^1,\degion^2,\degion^3$ that have empty intersection.

\medskip

\section{\bf Proof of Lemma \ref{lem:map graph}} \label{app:proof stokes}
We will use induction on $s$ to establish the following sequence of statements, for $s=3,4,\ldots$:
\begin{equation*}
H_s:\, \textrm{``Every fine cycle contained in a Euclidean ball of radius } s\delta \textrm{,  has zero weight''}.
\end{equation*}

\noindent
{\bf Base case} ($s=3$){:\bf}
Consider a fine cycle $x_1,\ldots, x_t$, and suppose that it is contained in the closed $3\delta$\pur{-}neighbourhood of a point $y\in\R^n$. 
Let $\degion_{k_0}$ be a region that contains $y$.
Since $3\delta=\gamma$, Lemma~\ref{lem:prop of tiling of Rn} implies that $\degion_{k_0}\bigcap\degion_{k_i}\bigcap\degion_{k_{i+1}}$ is nonempty, for $i=1,\ldots, t-1$.
Then, it follows from Lemma~\ref{lem:D is curl free} that
\begin{eqnarray}
b_{k_0k_i} +  b_{k_ik_0} &=& 0,\quad i=1,\ldots, n-1,   \label{eq:2 cycle with y =0 base}\\
 b_{k_0k_i} + b_{k_ik_{i+1}} + b_{k_{i+1}k_0} &=& 0,\quad i=1,\ldots, n-1. \label{eq:3 cycle with y =0 base}
\end{eqnarray} 
Therefore,
\begin{eqnarray*}
W\big(x_1,\ldots, x_t\big) &=& \sum_{i=1}^{t-1} b_{k_ik_{i+1}}\\
&=& \sum_{i=1}^{t-1} \Big( b_{k_0k_i} + b_{k_ik_{i+1}} + b_{k_ik_0}  \Big)\\
&=& \sum_{i=1}^{t-1} \Big( b_{k_0k_i} + b_{k_ik_{i+1}} + b_{k_{i+1}k_0}  \Big)\\
&=& 0,
\end{eqnarray*}
where the equalities follow from the definition of the weight $W(\cdot)$, \eqref{eq:2 cycle with y =0 base}, $x_n=x_1$, and \eqref{eq:3 cycle with y =0 base}, respectively.

\vspace{5pt}
\noindent 
{\bf Induction step:}
We now fix some $s\geq 4$ and assume the induction hypothesis that $H_{s-1}$ is true.
We will show that $H_s$ is also true, and
 do this by decomposing a fine cycle into fine cycles of smaller radii.
 Consider a fine cycle $x_1\ldots,x_t$ that is
 contained in a ball of radius $s\delta$.
Let $y$ be a point such that $\Ltwo{y-x_i}<(s+1)\delta$, for $i=1,\ldots,t-1$.
For any $i\le t-1$, consider a sequence $z_i^1,\ldots,z_i^s$ of equidistant points on the line segment between $y$ and $x_i$; see Fig.~\ref{fig:equidistant} for an illustration.
In view of the freedom in the choice of $y$, we can and will assume that $y$ is in the interior of some region $\degion_{k_0}$ and each $z_j^i$ is in the interior of a region $\degion_{k_j^i}$. 
Then, for any $i\le t-1$, the sequence 
$$\xi^i =\big(y,z_i^1,\ldots,z_i^s,x_i,x_{i+1},z_{i+1}^s, \ldots,z_{i+1}^1,y\big)$$ is a fine cycle.

\begin{figure*} 
\begin{center}
\includegraphics[width = .45\textwidth]{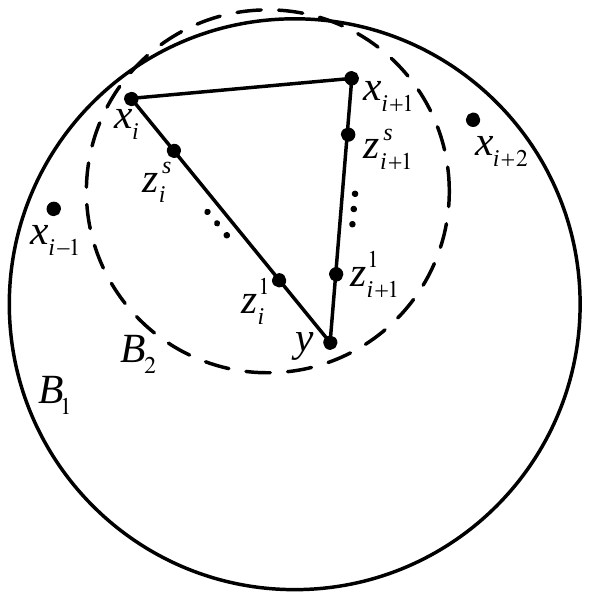}
\ifIEEE
\vspace{-1cm}
\fi
\end{center}
\caption{Decomposition of a fine cycle $x_1,\ldots,x_t$ into smaller fine cycles, in the proof of Lemma~\ref{lem:map graph}. Here the cycle $x_1,\ldots,x_t$ is contained in a ball $B_1$ of radius $s\delta$. The figure shows one of the cycles in the decomposition $y,z_i^1,\ldots,z_i^s,x_i,x_{i+1},z_{i+1}^s, \ldots,z_{i+1}^1,y$, which is contained in a smaller ball $B_2$, of radius at most $(s-1)\delta$. }
\label{fig:equidistant}
\end{figure*}

Moreover, the triangle with vertices $y$, $x_i$, and $x_{i+1}$ has two edges of length at most $(s+1)\delta$ and one edge, $x_ix_{i+1}$, of length at most $\delta$. Therefore, this triangle is contained in a closed ball of radius
\begin{equation} \label{eq:s+1+1/2  le s-1}
\frac{(s+1)\delta+\delta}{2} \,=\,\left( \frac{s }{2}+1\right)\,\delta \,\le\, (s-1)\delta,
\end{equation}
where the inequality holds because $s\ge 4$.
Then, the induction hypothesis $H_{s-1}$ implies that the cycle $\xi^i$ has zero weight.

Consider now a cycle $\xi$ which is the concatenation of the cycles $\xi^1,\ldots,\xi^{t-1}$.
The weight of $\xi$ equals the sum of the weights of the cycles $\xi^i$, and is therefore zero.
Finally, note that $\xi$ consists of the arcs of the original cycle $x_1,\ldots,x_{t-1}$, together with the intermediate paths that join $y$ to $x_i$, for the different $i$. For every $i$, the path from $y$ to $x_i$ is traversed twice, in opposite directions. 
Hence, using the property $b_{ij}+b_{ji}=0$ of the weights, the total contribution of these intermediate paths is zero. 
It follows that the weight of the original cycle $x_1,\ldots,x_{t-1}$ is equal to the weight of $\xi$, which is zero.

\medskip

\section{\bf Proof of Lemma \ref{lem:Phi is well defined and convex}} \label{subsec:phi is convex}

%
We consider the function $\Phit$ defined by $\Phit(x)=\max_{i} \, (- \mu_i^Tx+b_i )$.
We first prove the ``if'' direction of the result. That is, we fix some $x\in \degion_i$ and show that  $\Phit(x) =  -\mu_i^Tx+b_i$; equivalently, that
\begin{equation}\label{eq:if x in d then its a max}
-\mu_i^Tx+b_i \ge -\mu_{j}^Tx+b_{j},\qquad {\mbox{for all } x\in\degion_i \mbox{ and all }j.}
\end{equation}

Fix some $j\ne i$, and a point $y\in \degion_{j}$. 
Given that $\degion_j$ has nonempty interior, which is disjoint from $\degion_i$, we can and will assume that $y\neq x$.
By tracing the different regions crossed by the segment $xy$, it is not hard to see that we can find a finite sequence of \emph{distinct} points $x_1,\ldots,x_k$ and associated regions $\degion^{1},\ldots,\degion^{k}$ such that:
\begin{itemize}
\item[(i)]
 $x_1=x$ and $x_k=y$;
\item[(ii)] 
$x_{l}\in \degion^{l}$, for $l=1,\ldots,k$; in particular, 
 $\degion^{1}=\degion_i$ and $\degion^{k}=\degion_j$;
\item[(ii)]  the regions $\degion^{l}$ and $\degion^{l+1}$ are adjacent, for $l=1,\ldots,k-1$;  
\end{itemize}

We let $\mu^l$ and $b^l$ be the drift vector and weight, respectively, associated with region $\degion^l$. Then, 
for $l=1,\ldots,k-1$, we have $x_l\in\D^{l}$, and it follows from  \eqref{eq:bij is bi-bj} and \eqref{eq:def of bij in proof of th pwc} that
\begin{equation*}
\begin{split}
b^{l} - b^{l+1} \,  = \, \sup_{z\in\degion^{l}}\big(\mu^{l} -\mu^{l+1}  \big)^T z
\,\ge\,   \big(\mu^l -\mu^{l+1}  \big)^T x_l.
\end{split}
\end{equation*}
Equivalently, 
\begin{equation} \label{eq:j j+1 ineq}
-\big(\mu^{l} -\mu^{l+1}  \big)^T x_l + \big(b^{l} - b^{l+1}\big) \ge 0.
\end{equation}
Interchanging the roles of $i_l$ and $i_{l+1}$, we obtain
\begin{equation*}
-\big(\mu^{l+1} -\mu^l  \big)^T x_{l+1} 
+  \big(b^{l+1} - b^{l}\big) \ge 0.
\end{equation*}
Adding the last two inequalities, we get
\begin{equation*} 
-\big( \mu^{l}-\mu^{l+1}  \big)^T \big(x_l - x_{l+1}\big) \ge 0.
\end{equation*}
Since $x_{l+1}$ lies further along the line segment that connects $x_1$ to $x_k$, we see that
  $x_l - x_{l+1}$ is a \emph{positive} multiple of $x_1 - x_l$. Consequently, if $l<k$, then
\begin{equation} \label{eq:ij - ij+1 ineq}
-\big(\mu^{l} -\mu^{l+1}  \big)^T \big(x_1 - x_l\big) \ge 0.
\end{equation}
Adding \eqref{eq:j j+1 ineq} and \eqref{eq:ij - ij+1 ineq}, we obtain
\begin{equation*}
-\big(\mu^l -\mu^{l+1}  \big)^T x_1 + \big(b^{l} - b^{l+1}\big) \ge 0.
\end{equation*}
It follows that
\begin{equation*} 
\begin{split}
-\mu_i^Tx +b_i \,&=\, -(\mu^{1})^Tx_1 +b^{1} \\ 
& = \, \sum_{l=1}^{k-1} \Big(  -\big(\mu^{l} -\mu^{l+1}  \big)^T x_1 + \big(b^{l} - b^{l+1}\big) \Big)  -(\mu^{k})^Tx_1 +b^{k}\\
& \ge \, -(\mu^{k})^Tx_1 +b^{k}\\
& =\,  -\mu_j^Tx +b_j.
\end{split}
\end{equation*}
Recall that $x$ is an arbitrary element of $\degion_i$ and that $j$ is an arbitrary index. 
This establishes Eq.~\eqref{eq:if x in d then its a max} and concludes the proof of the ``if'' direction of the lemma.

We now establish the converse (``only if'') direction. 
We have already shown that for $x\in \degion_i$, $i$ is a maximizing index in the definition of $\Phit$. The proof of the converse starts by establishing a stronger statement: in the interior of $\degion_i$, the maximizing index is  unique. (The assumption that the $\mu_i$ are distinct will have to be invoked  here.)

%

\begin{claim} \label{claim:unique maximizer in the interior}
Fix some $x\in\R^n$. The index $i$ attains the maximum (over $j$) in the expression $\max_j \, (- \mu_j^Tx+b_j )$ and is the unique maximizer if and only if $x\in\degion_i^\circ$. 
\end{claim}

\begin{proof}[Proof of the Claim]
To establish the ``only if''  direction, suppose that $i$ is the unique maximizer of $\max_j \, (- \mu_j^Tx+b_j )$; equivalently
 $- \mu_j^Tx+b_j < \Phit(x)$, for all $j\ne i$. 
Then, Eq.~\eqref{eq:if x in d then its a max} implies that $x\not\in\degion_j$, for  all $j\ne i$.
Since $\bigcup_j \degion_j=\R^n$ and each $\degion_j$ is closed, it follows that $x\in \R^n \backslash \bigcup_{j\ne i} \degion_j = \degion_i^\circ$.

To establish the ``if'' direction, fix an index $i$ and a point $x\in\degion_i^\circ$. 
Take an arbitrary $j\ne i$ and a small enough $\epsilon$ such that $y\triangleq x+\epsilon (\mu_i-\mu_j)\in \degion_i$.
Then, 
\begin{equation}\label{eq:x in interior implies strict inequality}
\begin{split}
 -\mu_i^T x + b_i \,& =\, -\mu_i^T y + b_i + \epsilon\mu_i^T  (\mu_i-\mu_j) \\
&\ge\, -\mu_j^T y + b_j + \epsilon\mu_i^T  (\mu_i-\mu_j) \\
&=\,  -\mu_j^T x + b_j + \epsilon (\mu_i-\mu_j)^T  (\mu_i-\mu_j)\\
&=\,  -\mu_j^T x + b_j + \epsilon \Ltwo{\mu_i-\mu_j}^2\\
& >\, -\mu_j^T x + b_j ,
\end{split}
\end{equation}
where the first inequality is due to \eqref{eq:if x in d then its a max} after interchanging the indices $i$ and $j$, and the last inequality holds because the $\mu_i$ are distinct. 
Since Eq. \eqref{eq:x in interior implies strict inequality} holds for every $j\ne i$, it follows that 
$i$ is the unique maximizer of $\max_j \, (- \mu_j^Tx+b_j )$, and the claim follows.
\end{proof}

Suppose now that   $\Phit(x) =  -\mu_i^Tx+b_i$; in particular, 
\begin{equation}\label{eq:d1copy}
-\mu_i^Tx+b_i \geq
-\mu_j^Tx+b_j,
\end{equation}
for all $j$.
We need to show that $x\in \degion_i$.
From the definition of a polyhedral tiling, $\degion_i$ has nonempty interior. 
Fix an arbitrary $y\in \degion_i^\circ$.
Claim \ref{claim:unique maximizer in the interior} implies that  for any $j\ne i$,
\begin{equation}\label{eq: i j inequality at y}
 -\mu_i^T y + b_i > \, -\mu_j^T y + b_j .
\end{equation}
For $\epsilon\in(0,1)$, let $z_\epsilon \triangleq (1-\epsilon)x+\epsilon y$. 
It follows from Eqs.~\eqref{eq:d1copy} and \eqref{eq: i j inequality at y} that for any $\epsilon\in(0,1)$,
\begin{equation}\label{eq: i j inequality at z}
 -\mu_i^T z_\epsilon + b_i > \, -\mu_j^T z_\epsilon + b_j , \qquad \forall\  j\ne i.
\end{equation}
Employing Claim~\ref{claim:unique maximizer in the interior} once again, it follows from \eqref{eq: i j inequality at z} that for any $\epsilon\in(0,1)$, $z_\epsilon \in \degion_i^\circ$. 
Thus, as $\epsilon$ goes to zero,  $z_\epsilon$ converges to $x$ from within the interior of $\degion_i$.
Since $\degion_i$ is closed, we conclude that $x\in\degion_i$, and the lemma follows.


\section{Proof of the converse direction of Part (a) of Theorem~\ref{th:pwc}}
\label{sec:easy-proof}
Consider an FPCS system, with $F(x)=-\partial \Phit(x)$, where
$$\Phit(x)=\max_{i=1,\ldots,m} \, (- \mu_i^Tx+b_i ),$$
and where the pairs $(\mu_i,b_i)$ are distinct.

For any $i$, let $\degion_i=\{x\mid \Phit(x) = - \mu_i^Tx+b_i\}$, which is a polyhedron.
Clearly, the union of the sets $\degion_i$ is all of $\R^n$. 
Let $\I$ be the set of all $i$ for which the polyhedron $\degion_i$ has nonempty interior.
Then, $\R^n \backslash \bigcup_{i\in\I} \degion_i$ is an open subset of the empty interior set $ \bigcup_{i\not\in\I} \degion_i$, and is therefore empty. 
It follows that $ \bigcup_{i\in\I} \degion_i = \R^n$. 
Also note that for any two regions $\degion_i$ and $\degion_j$, their interiors $\degion_i^\circ = \{x\mid - \mu_i^Tx+b_i >- \mu_k^Tx+b_k ,\, \forall k\ne i \}$ and $\degion_j^\circ = \{x\mid - \mu_j^Tx+b_j >- \mu_k^Tx+b_k ,\, \forall k\ne j \}$ are disjoint. 
This proves that the sets $\degion_i$, $i\in \I$, comprise a polyhedral tiling.

Let $\Phit'(x)=\max_{i\in I} \, (- \mu_i^Tx+b_i )$. Then, $\Phit'(x)\le\Phit(x)$, for all $x$. On the other hand, since $ \bigcup_{i\in\I} \degion_i = \R^n$, any $x\in\R^n$ lies in $\degion_i$ for some $i\in \I$. Thus, for $x\in\degion_i$, we have $\Phit(x) =  - \mu_i^Tx+b_i \le \max_{j\in \I} \, (- \mu_j^Tx+b_j ) = \Phit'(x)$. We conclude that, $\Phit'(x) = \Phit(x)$, for all $x\in\R^n$.

Note that if $\mu_j=\mu_i$ and
(without loss of generality) $b_i<b_j$, then $-\mu_i^Tx+b_i < -\mu^T_j x+b_j$ for all $x$, so that $i$ cannot attain the maximum in the definition of $\Phi$, and $i\notin\degion_i$. This shows that when we restrict to indices $i$ in $\I$, the corresponding vectors $\mu_i$ must be distinct, as required in the definition of polyhedral hybrid systems.

Finally, it follows form the subdifferential formula for the pointwise maximum of linear functions \cite[Section 3.1.1]{Bert15} that $-\partial \Phit(x) = -\partial \Phit'(x) = {\rm Conv}\big(\{\mu_i \mid x\in\degion_i, \,i\in \I\}\big)$. Therefore, $-\partial \Phit(x)$ is the polyhedral hybrid system associated with the regions $\degion_i$ and the vectors $\mu_i$, for $i\in\I$.


\section{\bf Proof of Lemma \ref{lem:Ri closed}} \label{sec:proof lem Ri closed}
Let us consider a region $\region_i$.
We will prove the equivalent statement that the complement of $\region_i$ is open.
 Let us fix some
$x\not\in\region_i$. 
From the definition of $\region_i$, the assumption $x\not\in\region_i$ implies that $\mu_i\not\in F(x)$. 
Moreover, since $F$ is well-formed, $F(x)$ is  closed.  
Therefore, there exists an open subset $\mathcal{V}$ of $\R^n$ such that $F(x)\subseteq \mathcal{V}$ and $\mu_i\not\in \mathcal{V}$. 
It then follows from the upper-semicontinuity of $F(\cdot)$ that there exists an open neighbourhood $\mathcal{U}$ of $x$ such that $F(y)\subseteq \mathcal{V}$, for all $y\in \mathcal{U}$.
Therefore, $\mu_i\not\in  F(y)$, for all $y\in \mathcal{U}$.
Equivalently, $y\notin\region_i$, for all $y\in  \mathcal{U}$. Hence, for any $x$ in the complement of $\region_i$, there exists an open neighborhood of $x$ which is contained in the complement of $\region_i$. This shows that the complement of $\region_i$ is open, and concludes the proof.


\section{\bf Proof of Lemma \ref{lem:Di is finite seg}} \label{subsec:proof Di is finite seg}
We will show that the regions  $\degion_i\triangleq\textrm{closure}\big(\region_i^\circ\big)$,  $i=1,\ldots,m$, defined in \eqref{eq:def degion i in the proof of th pwc}, 
and restricting to those regions that are nonempty (i.e., with $i\in\I$)
form a tiling. 

From the finite-partition property,  for every $x$ there exists some $i$ such that $\mu_i\in F(x)$.
It follows that
every $x$ must belong to some $\region_i$
(cf.~Equation \eqref{eq:def region i}), and $\bigcup_i \region_i =\R^n$. Thus,
\begin{equation}\label{eq:g1}
\begin{split}
\R^n  \,\backslash\Big( \,\bigcup_i  \degion_i \Big)\, & =\, \Big(\bigcup_i \region_i   
\Big) \,\backslash \,   \Big(\bigcup_i \degion_i\Big)\\
&\subseteq \,\bigcup_i \big(\region_i  \,\backslash \,  \degion_i  \big)\\
&\subseteq \,\bigcup_i \big(\region_i  \,\backslash \, \region_i^\circ  \big).
\end{split}
\end{equation}
Each $\region_i \, \backslash\, \region_i^\circ$ has empty interior. It follows that
$\bigcup_i \big(\region_i  \,\backslash \, \region_i^\circ  \big)$ 
has empty interior as well. 
From the definition $\degion_i  = \textrm{closure}\big(\region_i^\circ\big)$, each set $\degion_i$ is automatically closed.
Hence $\R^n  \backslash \big(\bigcup_i  \degion_i\big)$ is open. At the same time, from \eqref{eq:g1}, the latter set is  contained  in the empty-interior set  $\bigcup_i \big(\region_i  \,\backslash \, \region_i^\circ  \big)$. This implies that $\R^n  \backslash \big(\bigcup_i  \degion_i\big)$ is empty, and therefore
\begin{equation}\label{eq:G2}
\bigcup_i \degion_i = \R^n.
\end{equation}




Note that for every $i$, $\degion_i\supseteq \region_i^\circ$. Thus, any interior point of $\region_i$ is also an interior point of $\degion_i$. In particular, if $\degion_i$ has empty interior, then $\region_i^\circ$ is empty, and so is its closure, $\degion_i$. It follows that whenever $\degion_i$ is nonempty (formally, when $i\in\I$), $\degion_i$ has nonempty interior. 
From now on, we restrict attention to nonempty regions, $\degion_i$, i.e., with $i\in\I$.

If there is a single nonempty region $\degion_i$, then that region is all of $\R^n$ and is therefore closed and convex. Suppose now that there are multiple nonempty regions, and let us fix some distinct $i,j\in\I$. 
 Let $x_i\in\degion_i^\circ$ and $x_j\in\degion_j^\circ$, and let $x_i(\cdot)$  and $x_j(\cdot)$ be two trajectories with initial conditions $x_i(0)=x_i$ and $x_j(0)=x_j$.  
It can be seen that for small positive times $t$,  $x_i+\mu_i t$ is a possible trajectory; in fact, because of the nonexpansive property, it is the unique trajectory. It follows that at time $t=0$, we have $\dot{x}_i(0) = \mu_i$ and, similarly, $\dot{x}_j({0}) = \mu_j$. 
Therefore,
\begin{equation*} 
\begin{split}
\big(\mu_i-\mu_j\big)^T\big(x_i-x_j\big) \, &=\, \big(\dot{x}_i(0) -\dot{x}_j(0) \big)^T\big(x_i(0)-x_j(0)\big)\\
&  =\, \frac{1}{2}\cdot\frac{d^+}{dt} \Ltwo{x_i(t)-x_j(t)}^2 \Big|_{t=0}.
\end{split}
\end{equation*}
Using the nonexpansive property of the dynamics, we conclude that
\begin{equation} \label{eq:mui -muj xi - xj le 0}
\big(\mu_i-\mu_j\big)^T\big(x_i-x_j\big) \, \le\, 0.
\end{equation}
Since $\degion_i = \textrm{closure}\big(\region_i^\circ\big) 
	= \textrm{closure}\big(\degion_i^\circ\big)$,
for any $x_i\in\degion_i$, there exists a sequence of points in the interior of $\degion_i$ that converges to $x_i$. 
 For any $x_i\in\degion_i$ and $x_j\in\degion_j$, the inequality \eqref{eq:mui -muj xi - xj le 0} holds along
such sequences of points converging to $x_i$ and $x_j$. 
Hence,  \eqref{eq:mui -muj xi - xj le 0} also holds for every $x_i\in\degion_i$ and $x_j\in\degion_j$.

Consider again some distinct $i,j\in\I$, and  let
\begin{equation} \label{eq:def bij for degion}
d_{ij} \,\triangleq\, \inf_{x_j\in \degion_j}  \big(\mu_i-\mu_j\big)^T x_j.
\end{equation}
Let us fix some $x_i\in \degion_i$. It follows from  \eqref{eq:mui -muj xi - xj le 0}   that   
as $x_j$ ranges over $\degion_j$, the expression 
$\big(\mu_i-\mu_j\big)^T x_j$  is lower bounded by $\big(\mu_i-\mu_j\big)^T x_i$. It follows that  the infimum in Eq.~\eqref{eq:def bij for degion}, and  therefore $d_{ij}$ as well
 is finite.  In particular, this justifies the second equality in the following calculation:
\begin{eqnarray} \label{eq:bij+bji ge0}
d_{ij} + d_{ji} &=&
\inf_{ x_j\in \degion_j} \big(\mu_i-\mu_j\big)^T x_j
+
\inf_{ x_i\in \degion_i} \big(\mu_j-\mu_i\big)^T x_i\\
&=&
\inf_{\substack{x_i\in \degion_i \\ x_j\in \degion_j}} \big(\mu_i-\mu_j\big)^T \big(x_j-x_i\big) \\ &
=& -\sup_{\substack{x_i\in \degion_i \\ x_j\in \degion_j}} \big(\mu_i-\mu_j\big)^T \big(x_i-x_j\big) \\
& \ge & 0,
\end{eqnarray}
where the inequality is due to \eqref{eq:mui -muj xi - xj le 0}.

For every $i\in\I$,  we define the polyhedron 
\begin{equation}
\pp_i = \bigcap_{j\ne i,\, j\in\I} \big\{ x\, \mid\,  \big(\mu_i-\mu_j\big)^T x \le  d_{ij} \big\}.
\end{equation}
Suppose that $x_i\in \degion_i$. Then, for every $j\in\I$, with $j\neq i$, and any $x_j\in\degion_j$,
\eqref{eq:mui -muj xi - xj le 0} yields
$(\mu_i-\mu_j)^T x_i \leq (\mu_i-\mu_j)^T x_j$. Taking the infimum over all $x_j\in\degion_j$, and using the definition of $d_{ij}$, we obtain $(\mu_i-\mu_j)^T x_i\leq d_{ij}$. Since this is true for every such $j$, and comparing with the definition of $\pp_i$, we obtain
\begin{equation} \label{eq:degion in pp}
\degion_i\subseteq \pp_i,
\end{equation}
for all $i$. 
It then follows from  \eqref{eq:G2}  that $\bigcup_i \pp_i = \R^n$.

\hide{
For distinct $i$ and $j$, it follows from \eqref{eq:def bij for degion} that if $x\in\pp_i^\circ$, then $\big(\mu_j-\mu_i\big)^T x > d_{ji}$. Therefore, for any $x\in\pp_i^\circ$,
\begin{equation}
\big(\mu_i-\mu_j\big)^T x  \, =\, -\big(\mu_j-\mu_i\big)^T x \,<\,- d_{ji}  \,\le\, d_{ij},
\end{equation}
where the last inequality is due to \eqref{eq:mui -muj xi - xj le 0}. 
\comm{I cannot follow the previous two sentences. How about the following.}}
Suppose again that $i,j\in \I$, with $j\neq i$.
If $x$ belongs to the interior of $\pp_i$, none of the inequality constraints in the definition of $\pp_i$ can be satisfied with equality, and therefore,
\begin{equation}\label{eq:mux<d}
\big(\mu_i-\mu_j\big)^T x  \,<\, d_{ij}.
\end{equation}
On the other hand, from \eqref{eq:def bij for degion}, if $x\in \pp_j$, then $\big(\mu_i-\mu_j\big)^T x  \ge d_{ij}$, which contradicts \eqref{eq:mux<d}. Therefore, 
\begin{equation}\label{eqref:disjoint pp}
\pp_i^\circ \bigcap \pp_j = \emptyset.
\end{equation}
We conclude that the polyhedra $\pp_i$ have disjoint interiors and comprise a polyhedral tiling. 

We now show that $\pp_i\subseteq\degion_i$.
We have
\begin{equation}
\pp_i^\circ  \,\subseteq\, \R^n \backslash   \bigcup_{j\ne i,\, j\in I} \pp_j 
\,\subseteq\,  \R^n \backslash   \bigcup_{j\ne i,\, j\in I} \degion_j  
\,\subseteq\, \degion_i,
\end{equation}
where the inclusions follow from \eqref{eqref:disjoint pp}, \eqref{eq:degion in pp}, and \eqref{eq:G2}, respectively.
Therefore, 
\begin{equation}\label{eq:pp in degion}
\pp_i\, =\, \textrm{closure} (\pp_i^\circ)  \,\subseteq\, \textrm{closure} (\degion_i) \,=\, \degion_i.
\end{equation}
It then follows from \eqref{eq:pp in degion} and \eqref{eq:degion in pp} that $\degion_i=\pp_i$, for all $i$, and the regions $\degion_i$ thereby comprise a polyhedral tiling.

\end{document}